\documentclass[conference]{IEEEtran}
\IEEEoverridecommandlockouts
\usepackage{cite}
\usepackage{amsmath,amssymb,amsfonts}
\usepackage{graphicx}
\usepackage{textcomp}
\usepackage{xcolor}
\usepackage{multirow}
\usepackage{amsthm}
\usepackage{algorithmic}
\usepackage{algorithm}
\newtheorem{definition}{Definition}
\newtheorem{theorem}{Theorem}
\newtheorem{lemma}{Lemma}
\DeclareMathOperator*{\argmax}{arg\,max}

\usepackage{url}
\def\BibTeX{{\rm B\kern-.05em{\sc i\kern-.025em b}\kern-.08em
    T\kern-.1667em\lower.7ex\hbox{E}\kern-.125emX}}
\begin{document}

\title{Differentially Private Selection\\ using Smooth Sensitivity
\thanks{Preprint of an article accepted at IEEE IPCCC 2024.}
\thanks{This work was supported by JSPS KAKENHI Grant Numbers 20H05967, 21H05052, 23K28035, and 23K18501 and JSPS Grant-in-Aid for JSPS Fellows Grant Number 23KJ0649.
}
}


\author{\IEEEauthorblockN{Akito Yamamoto}
\IEEEauthorblockA{\textit{Human Genome Center, The Institute of Medical Science,}\\ \textit{The University of Tokyo} \\
Tokyo, Japan \\
a-ymmt@ims.u-tokyo.ac.jp}
\and
\IEEEauthorblockN{Tetsuo Shibuya}
\IEEEauthorblockA{\textit{Human Genome Center, The Institute of Medical Science,}\\ \textit{The University of Tokyo} \\
Tokyo, Japan \\
tshibuya@hgc.jp}}

\maketitle

\begin{abstract}
With the growing volume of data in society, the need for privacy protection in data analysis also rises. In particular, private selection tasks, wherein the most important information is retrieved under differential privacy are emphasized in a wide range of contexts, including machine learning and medical statistical analysis. However, existing mechanisms use {\it global sensitivity}, which may add larger amount of perturbation than is necessary. Therefore, this study proposes a novel mechanism for differentially private selection using the concept of {\it smooth sensitivity} and presents theoretical proofs of strict privacy guarantees. Simultaneously, given that the current state-of-the-art algorithm using {\it smooth sensitivity} is still of limited use, and that the theoretical analysis of the basic properties of the noise distributions are not yet rigorous, we present fundamental theorems to improve upon them. Furthermore, new theorems are proposed for efficient noise generation. Experiments demonstrate that the proposed mechanism can provide higher accuracy than the existing {\it global sensitivity}-based methods. Finally, we show key directions for further theoretical development. Overall, this study can be an important foundational work for expanding the potential of {\it smooth sensitivity} in privacy-preserving data analysis. The Python implementation of our experiments and supplemental results are available at \url{https://github.com/ay0408/Smooth-Private-Selection}.
\end{abstract}


\section{Introduction}

As the data flows in society increase, the theory and mechanisms for preserving privacy must be further enhanced in data analysis. In particular, retrieving the most important information from sensitive data while protecting privacy has been emphasized in various fields, including machine learning \cite{11,13,14} and bioinformatics \cite{17,18}. For example, in the context of genome statistical analysis, the private extraction of the positions in the human genome that are most relevant to a disease is extremely important. In such private selection tasks, the concept of differential privacy \cite{1} has been widely utilized; however, further rigorous exploration of its fundamental mechanisms is essential.

All existing methods for differentially private selection use the {\it global sensitivity} of a function \cite{7}; however, this considers the worst-case scenario and may lead to the addition of excessive noise for actual data analysis. A concept that can address this problem is {\it smooth sensitivity} \cite{3}, and in fact, in the context of statistics publication, it has been pointed out that it can provide significantly higher accuracy than {\it global sensitivity} \cite{5}. However, algorithms using {\it smooth sensitivity} have only been developed for numeric queries \cite{3,6}; there is still no mechanism for private selection tasks. In addition, from a fundamental perspective, the existing algorithm for numeric queries can be further deepened, and we can design more flexible algorithms. Moreover, there are remaining issues such as the lack of theorems on efficient noise generation and the absence of rigorous theoretical analysis of the basic properties of noise distributions.

Therefore, in this study, we propose a novel mechanism for differentially private selection using the concept of {\it smooth sensitivity}, along with fundamental theorems for the enrichment of the basic theory. The main contributions of this study are as follows:
\begin{enumerate}
\item We first enhance the state-of-the-art algorithm for numeric queries and present a new theorem for constructing an $\epsilon$-differentially private algorithm more flexibly. This expands the possibilities for algorithms using {\it smooth sensitivity}.

\item With the idea in our enhanced algorithm, we propose the first mechanism for differentially private selection using the concept of {\it smooth sensitivity}, the Smooth Private Selection. Simultaneously, while summarizing the basic definitions regarding {\it smooth sensitivity} in the context of private selection, we also present theoretical proofs of the proposed mechanism's strict privacy guarantees.

\item We then offer new fundamental theorems for efficient noise generation, followed by a more rigorous theoretical analysis of the basic properties of noise distributions than the existing study \cite{3}. This will expand the applicability of differentially private mechanisms using {\it smooth sensitivity}, including our smooth private selection.

\item We experimentally evaluate the utility of the smooth private selection to demonstrate that it can provide higher accuracy than existing {\it global sensitivity}-based methods, using genome statistics that are emphasized in genomic and medical data analysis. Specifically, our mechanism retrieves the most important information with higher probability while satisfying $\epsilon$-differential privacy. 
\end{enumerate}

The omitted proofs and the Python codes used in our experiments are available at \url{https://github.com/ay0408/Smooth-Private-Selection}.

\section{Related Work}

In this section, we briefly summarize the limitations of related studies and clarify the differences between them and this study.

\subsection{Smooth Sensitivity}

The concept of {\it smooth sensitivity} \cite{3} has been widely applied and studied in machine learning and statistical analysis \cite{6,23,24,25}. In particular, it has been suggested that the noise could be reduced compared to using {\it global sensitivity} \cite{6}, and several theorems and methods have been proposed to improve the basic algorithm and efficiently obtain {\it smooth sensitivity} \cite{5}. However, they remain of limited use, and enhancements to the theoretical foundation, including rigorous analysis of the noise distributions are still required. In this study, we address these issues and enrich the fundamental theorems and algorithms regarding {\it smooth sensitivity} as shown by contributions 1 and 3.

\subsection{Differentially Private Selection}

Private selection for retrieving important information while protecting privacy is important in data mining \cite{8} and genomic statistical analysis \cite{13}. All existing methods for differentially private selection use {\it global sensitivity} \cite{7,11}; no mechanism currently uses {\it smooth sensitivity}. We propose the first such mechanism and experimentally show that it can provide higher accuracy than {\it global sensitivity}-based methods while providing theoretical proofs of the privacy guarantees, as shown by contributions 2 and 4. In this study, we consider the pure $\epsilon$-differential privacy setting and do not discuss unknown domain settings such as for stream data \cite{10} or under $(\epsilon, \delta)$-differential privacy \cite{4}. Those aspects are important and will be considered for the further development of our mechanism.

\section{Preliminaries}

In this section, we first review the key concepts regarding differential privacy and {\it smooth sensitivity}, along with detailed issues in existing studies. Thereafter, we summarize the existing mechanisms for differentially private selection and clarify the aim of this study.

\subsection{Differential Privacy}

Differential privacy \cite{1} is a framework that emerged in cryptography to assess the privacy level of an individual's participation information in a dataset. In recent years, it has been widely applied in the context of data management \cite{21,28} and medical data analysis \cite{22}. This concept considers two {\it neighboring} datasets that differ by only one element and aims to create a condition where these datasets are almost indistinguishable. When $d(x,y)$ represents the Hamming distance between $x$ and $y$, that is, $d(x,y) = |\{i:x_i \neq y_i\}|$, the definition of $\epsilon$-differential privacy is as follows:

\begin{definition}
$($$\epsilon$-Differential Privacy \cite{2}$)$ 
\\
A randomized mechanism $M$ is $\epsilon$-differentially private if, for all neighboring datasets $x, y \in D^n$ satisfying $d(x,y) = 1$ and for any $\mathcal{S} \subset range(M)$, 
\begin{eqnarray}
\Pr[M(x)\in \mathcal{S}] \leq e^\epsilon \cdot \Pr[M(y) \in \mathcal{S}]. \nonumber
\end{eqnarray}
\end{definition}

$D^n$ represents the set of all possible datasets containing information on $n$ elements. The parameter $\epsilon\ (> 0)$ evaluates the privacy level, with smaller values representing stronger privacy guarantees. 

The most basic method for satisfying $\epsilon$-differential privacy is to use the {\it global\ sensitivity} of a function. The definition of {\it global\ sensitivity} is as follows:

\begin{definition}
$($Global Sensitivity \cite{2}$)$
\\
The global sensitivity of a function $f: D^n \rightarrow  \mathbb{R}^d$ is 
\begin{eqnarray}
GS_f = \max_{x,y:d(x,y)=1}\|f(x) - f(y)\|_1. \nonumber
\end{eqnarray}
\end{definition}

The {\it global sensitivity} is a constant value that does not depend on datasets. However, {\it global sensitivity} considers all possible {\it neighboring} datasets, and the value is often based on extremely rare situations in practical cases. Notably, in the context of statistics publication, the addition of more noise than necessary makes achieving high accuracy difficult \cite{5}. The concept of {\it smooth sensitivity} \cite{3} contributes to solving this problem.

\subsubsection{Smooth Sensitivity}

The concept of {\it smooth sensitivity} is based on {\it local sensitivity} \cite{3,27}, which represents the maximum change with respect to an input dataset. The definition is as follows:

\begin{definition}
$($Local Sensitivity \cite{3}$)$
\\
For a function $f: D^n \rightarrow \mathbb{R}^d$ and a dataset $x \in D^n$, the local sensitivity of $f$ at $x$ is
\begin{eqnarray}
LS_f(x) = \max_{y:d(x,y)=1} \|f(x)-f(y)\|_1. \nonumber
\end{eqnarray}
\end{definition}

In contrast to {\it global sensitivity}, {\it local sensitivity} depends on the relevant dataset and easy to compute, and $\forall x: LS_f(x) \leq GS_f$ holds. Because the noise scale also varies among different datasets, the following concept of {\it smooth upper bound} \cite{3} is necessary to satisfy differential privacy.

\begin{definition}
$($$\beta$-Smooth Upper Bound on LS \cite{3}$)$
\\
For $\beta > 0$, a function $S : D^n \rightarrow \mathbb{R}$ is a $\beta$-smooth upper bound on the local sensitivity of $f$ if it satisfies the following requirements:
\begin{eqnarray}
\forall x \in D^n: && S(x) \geq LS_f(x); \nonumber\\
\forall x,y \in D^n, d(x,y)=1: && S(x) \leq e^{\beta} \cdot S(y). \nonumber 
\end{eqnarray}
\end{definition}

A special case of {\it smooth upper bound} is {\it smooth sensitivity}, which is defined by

\begin{definition}
$($$\beta$-Smooth Sensitivity \cite{3}$)$
\\
For $\beta > 0$, the $\beta$-smooth sensitivity of $f$ at $x$ is
\begin{eqnarray}
S_{f,\beta}^*(x) = \max_{y \in D^n} \left( LS_f(y) \cdot e^{-\beta \cdot d(y,x)} \right). \nonumber
\end{eqnarray}
\end{definition}

For every $\beta$-{\it smooth upper bound} $S$, $\forall x: S_{f,\beta}^*(x) \leq S(x)$ holds \cite{3}. The larger $\beta$ is, the closer $S_{f,\beta}^*(x)$ is to $LS_f(x)$, and the computation time can be reduced \cite{5}. However, there is still no general method for efficient computation when $\beta$ is small; therefore, we present new theorems in Section IV.C.

When using {\it smooth sensitivity} or {\it smooth upper bound}, noise distributions to satisfy differential privacy include $(\alpha, \beta)$-{\it admissible} distributions. The definition of $(\alpha, \beta)$-{\it admissible} under $\epsilon$-differential privacy is as follows:

\begin{definition}
$($$(\alpha, \beta)$-Admissible \cite{3}$)$
\\
A probability distribution on $\mathbb{R}^d$, given by a density function $h$, is $(\alpha, \beta)$-admissible if, for $\alpha = \alpha(\epsilon)$, $\beta = \beta(\epsilon)$, the following two conditions hold for all $\Delta \in \mathbb{R}^d$ and $\lambda \in \mathbb{R}$ satisfying $\|\Delta\|_1 \leq \alpha$ and $|\lambda| \leq \beta$ and for all measurable subsets $\mathcal{S} \subseteq \mathbb{R}^d$:
\begin{eqnarray}
Sliding \ Property: \hspace{-0.4cm }&& \Pr_{Z \sim h} [Z \in \mathcal{S}] \leq e^{\frac{\epsilon}{2}} \cdot \Pr_{Z \sim h} [Z \in \mathcal{S} + \Delta]; \nonumber\\
Dilation \ Property: \hspace{-0.4cm} && \Pr_{Z \sim h} [Z \in \mathcal{S}] \leq e^{\frac{\epsilon}{2}} \cdot \Pr_{Z \sim h} [Z \in e^{\lambda} \cdot \mathcal{S}]. \nonumber
\end{eqnarray}
\end{definition}

Using an $(\alpha, \beta)$-{\it admissible} distribution, we can construct an $\epsilon$-differentially private algorithm for numeric queries from the following lemma. 

\begin{lemma}
$($$\epsilon$-differentially private algorithm using a smooth upper bound \cite{3}$)$
\\
Let $h$ be an $(\alpha,\beta)$-admissible noise probability density function, and set $\alpha = \alpha(\epsilon)$ and $\beta = \beta(\epsilon)$. Let $Z$ be a random variable derived from $h$. For a function $f: D^n \rightarrow \mathbb{R}^d$, let $S: D^n \rightarrow \mathbb{R}$ be a $\beta$-smooth upper bound on the local sensitivity of $f$. Thereafter, the algorithm $A(x) = f(x) + \frac{S(x)}{\alpha} \cdot Z$ is $\epsilon$-differentially private.
\end{lemma}

In this lemma, $S(x)$ and $\alpha$ can be considered to have a privacy budget of $\epsilon/2$ each (with reference to Definition 6). However, it is not necessary to distribute the same budget, and accuracy could be improved by distributing more to one of them. For example, if the effect of reducing $S(x)$ by increasing $\beta$ is greater than the effect of increasing $\alpha$, it would be preferable to assign a greater privacy budget to $\beta$, that is, the numerator of $\frac{S(x)}{\alpha}$. Based on this discussion, we extend Lemma 1 and present a new fundamental theorem to construct a more flexible algorithm in Section IV.A.

In Section IV.D, we also provide a detailed analysis of noise distributions. The $(\alpha,\beta)$-{\it admissible} distributions mentioned in the literature for satisfying $\epsilon$-differential privacy are the distributions with density $h(z) \propto \frac{1}{1 + |z|^\gamma}$ $(\gamma > 1)$ \cite{3,6}. Of these, for the standard Cauchy distribution with $\gamma = 2$, a rigorous analysis has been conducted of $(\alpha, \beta)$ \cite{5}; however, the other distributions still lack sufficient analysis. Therefore, we present a more rigorous $(\alpha, \beta)$ than existing studies for the distributions with general $\gamma > 1$. This enables the amount of noise to be further reduced.

\subsection{Differentially Private Selection}

Differentially private methods can be constructed for private selection tasks as well as for numeric queries. Differentially private selection has been studied in various areas of data mining and statistical analysis \cite{8,9,10,15}, and the exponential mechanism \cite{7} and permute-and-flip \cite{11} are state-of-the-art methods. In these methods, the desirability of each candidate being selected is represented as a score function, and $\epsilon$-differential privacy is satisfied by utilizing the {\it global sensitivity} of the function. The definition of {\it global sensitivity} for private selection is as follows:

\begin{definition}
$($Global Sensitivity for Private Selection \cite{7}$)$
\\
The global sensitivity of a score function $u: D^n \times \mathcal{R} \rightarrow  \mathbb{R}$ is 
\begin{eqnarray}
{GS}_{u,\mathcal{R}} = \max_{r \in \mathcal{R}}\max_{x,y: d(x,y) = 1} |u(x,r) - u(y,r)|, \nonumber
\end{eqnarray}
\end{definition}

As a score function, other than based on the original function $f$, several studies have proposed the SHD score representing the distance from a threshold \cite{15} and a function that utilizes the {\it local sensitivity} of $f$ \cite{16}.

\subsubsection{Exponential Mechanism}

The exponential mechanism \cite{7} is the most fundamental mechanism for differentially private selection and is defined as follows:
\begin{definition}
$($The Exponential Mechanism \cite{7}$)$
\\
The exponential mechanism selects and outputs a candidate $r \in \mathcal{R}$ with probability proportional to $\exp{\left( \frac{\epsilon \cdot u(x,r)}{2 \cdot GS_{u,\mathcal{R}}} \right)}$.
\end{definition}
The exponential mehanism can be realized by adding noise according to a Gumbel distribution and is identical to the mechanism $M_{EM}$ such that 
\begin{eqnarray}
M_{EM}(x) = \argmax_{r \in \mathcal{R}} \left\{ u(x,r) + \mathrm{Gb}_r \left(\frac{2 \cdot GS_{u,\mathcal{R}}}{\epsilon}\right) \right\}, \nonumber
\end{eqnarray}
where $\mathrm{Gb}_r(\eta)$ are i.i.d. samples from the Gumbel distribution with density function $h(z) = \frac{1}{\eta} \exp\left( - \frac{z}{\eta} - \exp\left( -\frac{z}{\eta} \right) \right)$ \cite{12}.

\subsubsection{Permute-and-Flip}

The outcome of the permute-and-flip is mathematically guaranteed never to be worse than that of the exponential mechanism \cite{11}. This procedure combines a random permutation and coin flips and can be constructed as Algorithm 1.

\begin{algorithm}
    \caption{The Permute-and-Flip \cite{11}}
    \begin{algorithmic}[1]
        \STATE $u_* = \max_{r \in \mathcal{R}} u(x,r)$
        \FOR{$r$ in $RandomPermutation(\mathcal{R})$}
        \STATE $p_r = \exp\left( \frac{\epsilon}{2 \cdot GS_{u,\mathcal{R}}} (u(x,r) - u_*) \right)$
        \IF{$Bernoulli(p_r)$}
        \RETURN $r$
        \ENDIF
        \ENDFOR
    \end{algorithmic}
\end{algorithm}

Similar to the exponential mechanism, the permute-and-flip can be realized by adding noise according to an exponential distribution and is identical to the mechanism $M_{PF}$ such that 
\begin{eqnarray}
M_{PF}(x) = \argmax_{r \in \mathcal{R}} \left\{ u(x,r) + \mathrm{Expo}_r \left( \frac{\epsilon}{2 \cdot GS_{u, \mathcal{R}}} \right) \right\}, \nonumber
\end{eqnarray}
where $\mathrm{Expo}_r(\lambda)$ are i.i.d. samples from the exponential distribution with density function $h(z) = \lambda e^{-\lambda z}\ (z \geq 0)$ \cite{12}.\\

There have been several studies on extending the above mechanisms \cite{13,14,16}; however, all of them are based on the {\it global sensitivity} of a score function, and there is still no method using {\it smooth sensitivity} for private selection. As in the case for numeric queries, {\it smooth sensitivity} can potentially reduce noise and significantly increase accuracy. Therefore, in this study, we propose a novel mechanism for private selection using {\it smooth sensitivity}, along with new fundamental theorems.

\section{Theoretical Contributions}

Herein, we first present a more flexible differentially private algorithm using {\it smooth upper bound}. Then, based on the discussion on our algorithm, we propose a novel method for private selection using the concept of {\it smooth sensitivity}, the Smooth Private Selection. Subsequently, we present new theorems on the efficient computation of {\it smooth sensitivity} and {\it smooth upper bound}. Furthermore, we perform a more detailed analysis of $(\alpha, \beta)$-{\it admissible} property of the noise distributions than existing studies, which is expected to expand the applicability of differentially private algorithms using {\it smooth sensitivity}.

\subsection{Enhanced Differentially Private Algorithm using Smooth Sensitivity}

By extending Lemma 1, we present in Theorem 1 a more general procedure for constructing an $\epsilon$-differentially private algorithm using a {\it smooth upper bound}.

\begin{theorem}
$($Enhanced $\epsilon$-differentially private algorithm using a smooth upper bound$)$
\\
Let $h$ be an $(\alpha,\beta)$-admissible noise probability density function, and set $\alpha' = \alpha(k \epsilon)$ and $\beta' = \beta( (2-k) \epsilon)$. Let $Z$ be a random variable derived from $h$. For a function $f: D^n \rightarrow \mathbb{R}^d$, let $S: D^n \rightarrow \mathbb{R}$ be a $\beta'$-smooth upper bound on the local sensitivity of $f$. Thereafter, the algorithm $A(x) = f(x) + \frac{S(x)}{\alpha'} \cdot Z$ is $\epsilon$-differentially private for any $k \in (0,2)$.
\end{theorem}

\begin{proof}
It is sufficient to show that, for all neighboring datasets $x,y \in D^n$ and for any $\mathcal{S} \subset range(A)$, 
\begin{eqnarray}
\Pr[A(x) \in \mathcal{S}] \leq e^\epsilon \cdot \Pr[A(y) \in \mathcal{S}]. \nonumber
\end{eqnarray}
Denote $\frac{S(x)}{\alpha'}$ by $N(x)$. Then, 
\begin{eqnarray}
\Pr[A(x) \in \mathcal{S}] &=& \Pr_{Z \sim h}[f(x) + N(x) \cdot Z \in \mathcal{S}] \nonumber\\
&=& \Pr_{Z \sim h}\left[Z \in \frac{\mathcal{S} - f(x)}{N(x)}\right] \nonumber\\
&\leq& e^{\frac{k \epsilon}{2}} \cdot \Pr_{Z \sim h}\left[Z \in \frac{\mathcal{S} - f(y)}{N(x)}\right] \nonumber\\
&\leq& e^{\frac{k \epsilon}{2}} \cdot e^{\frac{(2-k)\epsilon}{2}} \cdot \Pr_{Z \sim h}\left[Z \in \frac{\mathcal{S} - f(y)}{N(y)}\right] \nonumber\\
&=& e^\epsilon \cdot \Pr[A(y) \in \mathcal{S}] \nonumber
\end{eqnarray}
because 
\begin{eqnarray}
\frac{\|f(x) - f(y)\|_1}{N(x)} &=& \alpha' \cdot \frac{\|f(x) - f(y)\|_1}{S(x)} \nonumber\\
&\leq& \alpha' \cdot \frac{\|f(x) - f(y)\|_1}{LS_f(x)} \leq \alpha' \nonumber\\
&& \ \ \ \ \ \ \ \ \ \ \ \ \ \ \ \ \ \ \mathrm{(Sliding\ Property)}\nonumber\\
\mathrm{and} \ \ \ \ \left| \ln\frac{N(x)}{N(y)} \right| &=& \left| \ln\frac{S(x)}{S(y)} \right| \leq \beta' \nonumber\\
&& \ \ \ \ \ \ \ \ \ \ \ \ \ \ \ \ \ \ \mathrm{(Dilation\ Property)}. \nonumber
\end{eqnarray}
\end{proof}

The algorithm in Theorem 1 is superior to the existing one in Lemma 1 because the values of $\alpha'$ and $\beta'$ can be varied flexibly. In other words, when constructing an $\epsilon$-differentially private method using {\it smooth sensitivity}, the same privacy budgets need not be assigned to $\alpha$ and $\beta$, as mentioned in Section III.A. With this in mind, we propose a novel differentially private method using {\it smooth sensitivity} for private selection tasks.

\subsection{Differentially Private Selection using Smooth Sensitivity}

We first summarize the definitions of {\it local sensitivity}, {\it smooth upper bound}, and {\it smooth sensitivity} for private selection, with reference to existing studies \cite{3,16}. The discussion in this subsection is based on these definitions.

\begin{definition}
$($Local Sensitivity for Private Selection$)$
\\
For a score function $u: D^n \times \mathcal{R} \rightarrow  \mathbb{R}$ and a dataset $x \in D^n$, the local sensitivity of $u$ at $x$ for $r \in \mathcal{R}$ is
\begin{eqnarray}
{LS}_{u,\mathcal{R}}(x, r) = \max_{y: d(x,y) = 1} |u(x,r) - u(y,r)|, \nonumber
\end{eqnarray}
\end{definition}

\begin{definition}
$($$\beta$-Smooth Upper Bound on LS for Private Selection$)$
\\
For $\beta > 0$, a function $S: D^n \times \mathcal{R} \rightarrow \mathbb{R}$ is a $\beta$-smooth upper bound on the local sensitivity of a score function $u: D^n \times \mathcal{R} \rightarrow  \mathbb{R}$ if it satisfies the following requirements for any $r \in \mathcal{R}$:
\begin{eqnarray}
\forall x \in D^n: && S(x,r) \geq LS_{u,\mathcal{R}}(x,r); \nonumber\\
\forall x,y \in D^n, d(x,y)=1: && S(x,r) \leq e^{\beta} \cdot S(y,r). \nonumber 
\end{eqnarray}
\end{definition}

\begin{definition}
$($$\beta$-Smooth Sensitivity for Private Selection$)$
\\
For $\beta > 0$, the $\beta$-smooth sensitivity of a score function $u: D^n \times \mathcal{R} \rightarrow  \mathbb{R}$ at $x$ for $r \in \mathcal{R}$ is
\begin{eqnarray}
S_{u,\mathcal{R},\beta}^*(x,r) = \max_{y \in D^n} \left( LS_{u,\mathcal{R}}(y,r) \cdot e^{-\beta \cdot d(y,x)} \right). \nonumber
\end{eqnarray}
\end{definition} 

In this study, we consider two cases in which the noise distributions are respectively two-sided and one-sided. For example, the Gumbel distribution for the exponential mechanism is two-sided, and the exponential distribution for the permute-and-flip is one-sided. For both cases, we define a novel differentially private selection mechanism using the concept of {\it smooth sensitivity}, the Smooth Private Selection, as follows:
\begin{definition}
$($The Smooth Private Selection$)$
\\
Let $h$ be an $(\alpha, \beta)$-admissible noise probability density function, and set $\alpha' = \alpha(k \epsilon)$ and $\beta' = \beta(l \epsilon)$. Given a score function $u: D^n \times \mathcal{R} \rightarrow \mathbb{R}$, the smooth private selection mechanism is defined as 
\begin{eqnarray}
M_{SPS}(x) = \argmax_{r \in \mathcal{R}} \left\{ u(x,r) + \frac{\max_{r \in \mathcal{R}} S(x,r)}{\alpha'} \cdot Z_r \right\}, \nonumber
\end{eqnarray}
where $S$ is an $\beta'$-smooth upper bound on the local sensitivity of $u$, and $Z_r$ are i.i.d. random variables derived from $h$. Note that $k > 0$ and $l > 0$, and if $|M_{SPS}(x)| > 1$, return only the first element of the set. Furthermore, when $h(z) > h(z') > 0$ holds for all $z,z'$ satisfying $0 \leq z < z'$, $Z_r$ can be i.i.d. random variables derived from the density function $g$ such that
\begin{eqnarray}
g(z) = \begin{cases}
    0 \ \ \ \ \ \ (z < 0)\\
    \frac{h(z)}{\int_{w \geq 0} h(w) \mathrm{d}w} \ \ (z \geq 0)
\end{cases}. \nonumber
\end{eqnarray}
\end{definition}

The part of adding noise in this mechanism is inspired by our Theorem 1. Here, note that $h$ represents a two-sided distribution because of the sliding property in Definition 6, and $g$ represents a one-sided distribution. The achieved privacy level depends on whether the noise distribution is two-sided or one-sided. The following discussion details each case.

\subsubsection{Two-Sided Noise}

We first consider the case of two-sided noise distribution. The privacy guarantee of the smooth private selection for this case is provided by the following theorem:
\begin{theorem}
If $Z_r$ are derived from $h$, the smooth private selection mechanism satisfies $\left(\left( k + \frac{|\mathcal{R}|}{2} \cdot l \right) \cdot \epsilon\right)$-differential privacy.
\end{theorem}

\begin{proof}
For simplicity, we let $\mathcal{R} = \{1,2,\dots,m\}$ and $|\mathcal{R}| = m$. It is sufficient to show that, for all neighboring datasets $x,y \in D^n$,
\begin{eqnarray}
\Pr[M(x) = 1] \leq \exp\left( \left( k + \frac{m}{2} \cdot l \right) \cdot \epsilon \right) \cdot \Pr[M(y) = 1].
\end{eqnarray}
Here, we let $S(x) = \max_{r \in \mathcal{R}} S(x,r)$. Then, 
\begin{eqnarray}
\Pr[M(x) = 1] \hspace{-0.2cm} &=& \hspace{-0.2cm} \int_{v \in (-\infty, \infty)} \Pr\left[ u(x,1) + \frac{S(x)}{\alpha'} \cdot Z_1 = v \right] \nonumber\\
&& \ \ \ \ \ \ \cdot \Pr\left[ u(x,2) + \frac{S(x)}{\alpha'} \cdot Z_2 \leq v \right] \nonumber\\
&& \ \ \ \cdots \Pr\left[ u(x,m) + \frac{S(x)}{\alpha'} \cdot Z_m \leq v \right] \nonumber\\
&=& \hspace{-0.2cm} \int_{v \in (-\infty, \infty)} \Pr\left[ Z_1 = \frac{\alpha'(v-u(x,1))}{S(x)} \right] \nonumber\\
&& \ \ \ \ \ \ \cdot \Pr\left[ Z_2 \leq \frac{\alpha'(v-u(x,2))}{S(x)} \right] \nonumber\\
&& \ \ \ \cdots \Pr\left[ Z_m \leq \frac{\alpha'(v-u(x,m))}{S(x)} \right] .
\end{eqnarray}
Similarly,
\begin{eqnarray}
\Pr[M(y) = 1] \hspace{-0.2cm} &=& \hspace{-0.2cm} \int_{v \in (-\infty, \infty)} \Pr\left[ Z_1 = \frac{\alpha'(v-u(y,1))}{S(y)} \right] \nonumber\\
&& \ \ \ \ \ \ \cdot \Pr\left[ Z_2 \leq \frac{\alpha'(v-u(y,2))}{S(y)} \right] \nonumber\\
&& \ \ \ \cdots \Pr\left[ Z_m \leq \frac{\alpha'(v-u(y,m))}{S(y)} \right] . 
\end{eqnarray}
Because $\forall r \in \mathcal{R}: |u(x,r) - u(y,r)| \leq LS_{u,\mathcal{R}}(x,r) \leq S(x,r) \leq S(x)$, 
\begin{eqnarray}
(3) &\geq& \int_{v \in (-\infty, \infty)} \Pr\left[ Z_1 = \frac{\alpha'(v-u(y,1))}{S(y)} \right] \nonumber\\
&& \ \ \ \ \ \ \cdot \Pr\left[ Z_2 \leq \frac{\alpha'(v-u(x,2)-S(x))}{S(y)} \right] \nonumber\\
&& \ \ \ \cdots \Pr\left[ Z_m \leq \frac{\alpha'(v-u(x,m)-S(x))}{S(y)} \right]  \nonumber\\
&=& \int_{v \in (-\infty, \infty)} \Pr\left[ Z_1 = \frac{\alpha'(v-u(y,1)+S(x))}{S(y)} \right] \nonumber\\
&& \ \ \ \ \ \ \ \ \ \ \ \ \cdot \Pr\left[ Z_2 \leq \frac{\alpha'(v-u(x,2))}{S(y)} \right] \nonumber\\
&& \ \ \ \ \ \ \ \ \ \cdots \Pr\left[ Z_m \leq \frac{\alpha'(v-u(x,m))}{S(y)} \right] . 
\end{eqnarray}
Here, 
\begin{eqnarray}
&& -S(x) \leq u(x,1) - u(y,1) \leq S(x) \nonumber\\
&\iff& 0 \leq u(x,1) - u(y,1) + S(x) \leq 2 S(x), \nonumber
\end{eqnarray}
and
\begin{eqnarray}
\frac{S(x)}{S(y)} = \frac{\max_r S(x,r)}{\max_r S(y,r)} \leq e^{\beta'}. \nonumber
\end{eqnarray}
Therefore, 
\begin{eqnarray}
&& \Pr\left[ Z_1 = \frac{\alpha'(v-u(x,1))}{S(x)} \right] \nonumber\\
&\leq& e^{k \epsilon} \cdot \Pr\left[ Z_1 = \frac{\alpha'(v-u(y,1)+S(x))}{S(x)} \right] \nonumber\\
&\leq& e^{k \epsilon} \cdot e^{\frac{l \epsilon}{2}} \cdot \Pr\left[ Z_1 = \frac{\alpha'(v-u(y,1)+S(x))}{S(y)} \right] \nonumber
\end{eqnarray}
and $\forall r \in \{2,3,\dots,m\}:$
\begin{eqnarray}
&& \Pr\left[ Z_r \leq \frac{\alpha'(v-u(x,r))}{S(x)} \right] \nonumber\\
&\leq& e^{\frac{l \epsilon}{2}} \cdot \Pr\left[ Z_r \leq \frac{\alpha'(v-u(x,r))}{S(y)} \right]. \nonumber
\end{eqnarray}
Consequently, from (2) and (4), the relation (1) holds.
\end{proof}

From Theorem 2, our mechanism satisfies $\epsilon$-differential privacy when $k + \frac{|\mathcal{R}|}{2} \cdot l = 1$.

\subsubsection{One-Sided Noise}

Similar to the previous case, the privacy guarantee of our mechanism for the case of one-sided noise distribution is provided by the following theorem:

\begin{theorem}
If $Z_r$ are derived from $g$, the smooth private selection mechanism satisfies $\left(\left( k + \frac{|\mathcal{R}|-1}{2} \cdot l \right) \cdot \epsilon\right)$-differential privacy.
\end{theorem}

The proof is similar to that of Theorem 2 and is provided on our GitHub page. From Theorem 3, our mechanism satisfies $\epsilon$-differential privacy when $k + \frac{|\mathcal{R}|-1}{2} \cdot l = 1$. Compared to the previous case, the value of $\frac{S(x)}{k \cdot \alpha}$ can be made smaller, thereby reducing the amount of noise.

From the above discussion, when using the smooth private selection to satisfy $\epsilon$-differential privacy, the value of $l$ should be set reasonably small; that is, we must obtain a $\beta$-{\it smooth upper bound} for small $\beta$. However, no general method for doing this has yet been proposed, and it was almost impossible to compute in a practical time period mainly because of the large $|D^n|$. In the next subsection, therefore, we present new theorems for efficient computation.

\subsection{New Theorems for Obtaining Smooth Upper Bound}

We first present a new theorem for obtaining $\beta$-{\it smooth sensitivity} for small $\beta$ below a certain value determined by {\it global sensitivity} and {\it local sensitivity}, based on the original definitions (Definitions 2, 3, and 5).

\begin{theorem}
For any $\beta\ (> 0)$ satisfying 
\begin{eqnarray}
\beta \leq \min_{x: LS_f(x) \neq GS_f} \frac{1}{gd(x)} \cdot \ln\left( \frac{GS_f}{LS_f(x)} \right), 
\end{eqnarray}
where $gd(x) := \min_{y: LS_f(y) = GS_f} d(y,x)$ represents the shortest distance from $x$ to elements the local sensitivity of which equals global sensitivity, the following equality holds:
\begin{eqnarray}
\forall x: S_{f,\beta}^*(x) = GS_f \cdot e^{-\beta \cdot gd(x)}. \nonumber
\end{eqnarray}
\end{theorem}

\begin{proof}
From (5), for all $x: LS_f(x) \neq GS_f$,
\begin{eqnarray}
&& \beta \leq \frac{1}{gd(x)} \cdot \ln\left( \frac{GS_f}{LS_f(x)} \right) \nonumber\\
&\iff& LS_f(x) \leq GS_f \cdot e^{-\beta \cdot gd(x)}.
\end{eqnarray}
Let $Y$ be the set of $y$ satisfying $LS_f(y) = GS_f$ and $d(y,x) = gd(x)$. Because
\begin{eqnarray}
S_{f,\beta}^*(x) &=& \max_y \left( LS_f(y) \cdot e^{-\beta \cdot d(y,x)} \right) \nonumber\\
&=& \max\Bigl\{ GS_f \cdot e^{-\beta \cdot gd(x)}, \nonumber\\
&& \ \ \ \ \ \ \ \ \max_{y \notin Y} \left( LS_f(y) \cdot e^{-\beta \cdot d(y,x)} \right) \Bigr\}, \nonumber
\end{eqnarray}
it is sufficient to show that
\begin{eqnarray}
\forall x: \, GS_f \cdot e^{-\beta \cdot gd(x)} \geq \max_{y} \left( LS_f(y) \cdot e^{-\beta \cdot d(y,x)} \right). 
\end{eqnarray}

\hspace{-0.5cm} (I) When $LS(x) = GS_f$:

For any $y$, because $d(y,x) \geq gd(x) = 0$,
\begin{eqnarray}
GS_f \cdot e^{-\beta \cdot gd(x)} = GS_f \geq LS_f(y) \geq LS_f(y) \cdot e^{-\beta \cdot d(y,x)}. \nonumber
\end{eqnarray}

\hspace{-0.5cm} (II) When $LS_f(x) \neq GS_f$:

(a) For any $y: LS_f(y) = GS_f$, because $d(y,x) \geq gd(x)$,
\begin{eqnarray}
GS_f \cdot e^{-\beta \cdot gd(x)} \geq GS_f \cdot e^{-\beta \cdot d(y,x)} = LS_f(y) \cdot e^{-\beta \cdot d(y,x)}. \nonumber
\end{eqnarray}

(b) For any $y: LS_f(y) \neq GS_f$, where $z$ satisfies $LS_f(z) = GS_f$ and $d(z,y) = gd(y)$,
\begin{eqnarray}
GS_f \cdot e^{-\beta \cdot gd(x)} &\geq& GS_f \cdot e^{-\beta \cdot d(z,x)} \ \ \ \ [\because gd(x) \leq d(z,x)] \nonumber\\
&\geq& LS_f(y) \cdot e^{\beta \cdot d(z,y)} \cdot e^{-\beta \cdot d(z,x)} \ \ \ \ [\because (6)] \nonumber\\
&\geq& LS_f(y) \cdot e^{\beta \cdot d(z,y)} \cdot e^{-\beta (d(z,y) + d(y,x))} \nonumber\\
&& \ \ \ \ \ \ \ \ \ \ [\because d(z,x) \leq d(z,y) + d(y,x)] \nonumber\\
&=& LS_f(y) \cdot e^{-\beta \cdot d(y,x)}. \nonumber
\end{eqnarray}

From (I) and (II), the relation (7) holds. 
\end{proof}

The number of $x$ satisfying $LS_f(x) = GS_f$ is often $\mathcal{O}(1)$, and Theorem 4 eliminates the need to calculate {\it smooth sensitivity} in $\mathcal{O}(|D^n|)$ time based on the original Definition 5. We can also easily employ this theorem for private selection based on Definitions 7, 9, and 11.

One problem is that $\beta$ satisfying the inequality (1) may be too small; therefore, we present another theorem for obtaining a {\it smooth upper bound} for more practical $\beta$. In Theorem 5, for all $x$ with a {\it local sensitivity} close to $GS_f$, that is, for data that, in practical cases, are expected to be extremely rare, we set $S(x) = GS_f$. This expands the range of $\beta$ satisfying the inequality, and we can efficiently obtain smaller $S(x)$ for realistic data. The proof is similar to that of Theorem 4 and is provided on our GitHub page.

\begin{theorem}
Given a threshold $T$, we let the set of $x$ satisfying $LS_f(x) > T$ be $U$. For any $\beta\ (>0)$ satisfying 
\begin{eqnarray}
\beta \leq \min_{x \notin U} \frac{1}{ud(x)} \cdot \ln\left( \frac{GS_f}{LS_f(x)} \right), 
\end{eqnarray}
where $ud(x) := \min_{y \in U} d(y,x)$ represents the shortest distance between $x$ and $U$, the following function $S$ is a $\beta$-smooth upper bound:
\begin{eqnarray}
S(x) = GS_f \cdot e^{-\beta \cdot ud(x)} \ . \nonumber
\end{eqnarray}
\end{theorem}

By setting $T$ so that $|U| << |D^n|$, the {\it smooth upper bound} can be efficiently obtained for practical $\beta$ using Theorem 5. As with Theorem 4, Theorem 5 can also be employed for private selection based on Definitions 7, 9, and 10.

Developing more generalized theorems for obtaining {\it smooth sensitivity} and {\it smooth upper bound} is a challenging issue for the future work.

\subsection{Detailed Analysis of $(\alpha, \beta)$-Admissible Property}

In using our mechanism and other methods based on the concept of {\it smooth sensitivity}, it is critical to provide a rigorous $(\alpha, \beta)$-{\it admissible} property of the noise distribution for setting an appropriate $\beta$ for our Theorems 4 and 5 and for reducing the added noise. The discussion in the existing study \cite{3} are still insufficiently elaborate; therefore, we provide a new theorem for the distribution with density $h(z) \propto \frac{1}{1 + |z|^\gamma} (\gamma > 1)$, which can be used to satisfy $\epsilon$-differential privacy.

\begin{theorem}
For any $\gamma > 1$, the distribution with density $h(z) \propto \frac{1}{1 + |z|^\gamma}$ is 
\begin{eqnarray}
\ \ \ \ \ \left( \frac{\epsilon}{2 \cdot (\gamma - 1)^{\frac{\gamma - 1}{\gamma}}}, \, \frac{\epsilon}{2 (\gamma - 1)} \right)\mathchar`-{{\it admissible}}. \nonumber
\end{eqnarray}
\end{theorem}

\begin{proof}
We set
\begin{eqnarray}
\alpha = \frac{\epsilon}{2 \cdot (\gamma - 1)^{\frac{\gamma - 1}{\gamma}}}, \ \ \mathrm{and}\ \ 
\beta = \frac{\epsilon}{2 (\gamma - 1)}. \nonumber
\end{eqnarray}

First, we analyze the sliding property. From Definition 6, it is sufficient to show that $\ln \left( \frac{h(z)}{h(z+\Delta)} \right)$ is at most $\frac{\epsilon}{2}$ when $|\Delta| \leq \alpha$. Where $\phi(x) = \ln(1+x^\gamma)$, the following equalities hold:
\begin{eqnarray}
\ln\left( \frac{h(z)}{h(z+\Delta)} \right) &=& \ln\left( \frac{1 + |z + \Delta|^\gamma}{1 + |z|^\gamma} \right) \nonumber\\
&=& \phi(|z+\Delta|) - \phi(|z|). \nonumber
\end{eqnarray}
Here, there exists $\zeta > 0$ such that $\phi(|z+\Delta|) - \phi(|z|) \leq |\Delta| \cdot |\phi'(\zeta)|$. Because
\begin{eqnarray}
\phi'(\zeta) = \frac{\gamma \cdot \zeta^{\gamma-1}}{1 + \zeta^\gamma} = \frac{\gamma}{\zeta + \zeta^{1-\gamma}} \leq (\gamma - 1)^{\frac{\gamma-1}{\gamma}} \ \ [\because \zeta > 0], \nonumber
\end{eqnarray}
$\ln \left( \frac{h(z)}{h(z+\Delta)} \right) \leq |\Delta| \cdot (\gamma - 1)^{\frac{\gamma-1}{\gamma}}$. Thus, when $|\Delta| \leq \alpha$, $\ln \left( \frac{h(z)}{h(z+\Delta)} \right) \leq \frac{\epsilon}{2}$.

Thereafter, we analyze the dilation property. Similar to the above, it is sufficient to show that $\ln\left( \frac{h(z)}{e^{\lambda} \cdot h(e^\lambda z)} \right)$ is at most $\frac{\epsilon}{2}$ when $|\lambda| \leq \beta$. Because
\begin{eqnarray}
\ln\left( \frac{h(z)}{e^\lambda \cdot h(e^\lambda z)} \right) = \ln\left( \frac{1}{e^\lambda} \cdot \frac{1 + e^{\lambda \gamma} |z|^\gamma}{1 + |z|^\gamma} \right), \nonumber
\end{eqnarray}
when $\lambda \geq 0$,
\begin{eqnarray}
\ln\left( \frac{h(z)}{e^\lambda \cdot h(e^\lambda z)} \right) \leq \ln\left( \frac{1}{e^\lambda} \cdot \frac{e^{\lambda \gamma} + e^{\lambda \gamma} |z|^\gamma}{1 + |z|^\gamma} \right) = \lambda (\gamma - 1), \nonumber
\end{eqnarray}
and when $\lambda < 0$,
\begin{eqnarray}
\ln\left( \frac{h(z)}{e^\lambda \cdot h(e^\lambda z)} \right) \leq \ln\left( \frac{1}{e^\lambda} \right) = -\lambda = |\lambda|. \nonumber
\end{eqnarray}
Thus, when $|\lambda| \leq \beta$, $\ln\left( \frac{h(z)}{e^\lambda \cdot h(e^\lambda z)} \right) \leq \frac{\epsilon}{2}$. 
\end{proof}

The $(\alpha,\beta)$ in this theorem is larger than previously known $\left( (\alpha, \beta) = \left( \frac{\epsilon}{2(\gamma + 1)}, \frac{\epsilon}{2(\gamma + 1)} \right) \right)$ \cite{3}; therefore, we can reduce the added noise in {\it smooth sensitivity}-based methods including our proposed mechanism. 

In the next section, by utilizing Theorem 6 along with Theorems 4 and 5, we experimentally show that our smooth private selection can outperform existing {\it global sensitivity}-based methods.

\section{Experimental Evaluation}

We compared our smooth private selection with two existing methods based on the {\it global sensitivity}, the exponential mechanism and the permute-and-flip, thereby confirming the high utility of our mechanism. Specifically, we measured the accuracy, that is, the probability that the element selected by each mechanism is indeed the element with the highest score. For our mechanism, we also examined the effect of variation in $\gamma$ on the accuracy. 

We experimented with a transmission disequilibrium test (TDT), which is common in genomic statistical analysis, and considered a situation where the most significant element (single nucleotide polymorphism, SNP) is extracted based on the TDT statistics. This setting has also been addressed by existing studies \cite{17,18,5} in the bioinformatics field and is important in terms of privacy protection for personalized medicine. In particular, in the context of numeric queries, {\it smooth sensitivity} achieves significant higher accuracy than {\it global sensitivity} for the TDT statistics \cite{5}.

A TDT is used to examine linkages and correlations between SNPs and diseases in family-based studies. The most basic TDT statistic can be calculated from the following table, which represents a classification of $2N$ parents in $N$ families.

\begin{table}[htbt]
    \centering
    \vspace{-0.1cm}
    \begin{tabular}{cc|cc|c}
         &  & \multicolumn{2}{c|}{Non-Transmitted Allele} & \multirow{2}{*}{Total} \\ \cline{3-4}
         &  & \ \ \ \ $A_1$ & $A_2$ &  \\ \hline
        \multicolumn{1}{c|}{Transmitted} & $A_1$ & \ \ \ \ $a$ & $b$ & $a+b$ \\
        \multicolumn{1}{c|}{Allele} & $A_2$ & \ \ \ \ $c$ & $d$ & $c+d$ \\ \hline
        \multicolumn{2}{c|}{Total} & \ \ \ \ $a+c$ & $b+d$ & $2N$
    \end{tabular}
    \vspace{-0.1cm}
\end{table}

Using $b$ and $c$, the TDT statistic \cite{19} is calculated as $\chi^2_{TDT}(b,c) := \frac{(b-c)^2}{b+c}$. When $b = c = 0$, we set $\chi^2_{TDT}(0,0) = 0$. Under differentially private settings, we assume neighboring datasets differ by only one family. Then, the {\it global sensitivity} of TDT statistic is $8(N-1)/N$ \cite{17}.

In the experiments, we utilized Theorems 4 and 5 to compute the {\it smooth upper bound} $S(x)$ used in the smooth private selection. We set $T = 6$ for Theorem 5 and compute $gd(x)$ and $ud(x)$ based on the following Lemmas 2 and 3. 

\begin{lemma}
$(b,c)$ satisfies $LS_{\chi^2_{TDT}}((b,c)) > 6$ when
\vspace{-0.1cm}
\begin{eqnarray}
&& 0 \leq c < \frac{b-8}{7} \ \lor \ 2 \leq b < \frac{c+8}{7} \nonumber\\
&\lor& 0 \leq b < \frac{c-8}{7} \ \lor \ 2 \leq c < \frac{b+8}{7} \ . \nonumber
\end{eqnarray}
\end{lemma}

\begin{lemma}
The Hamming distance for TDT datasets can be
\vspace{-0.1cm}
\begin{eqnarray}
&& d(T(b,c), \, T(b',c')) \nonumber\\
&=& \begin{cases}
    \left\lceil \frac{|(b+c)-(b'+c')|}{2} \right\rceil \ \ \ \ \ \ \, ((b-b') \cdot (c-c') \geq 0)\\
    \vspace{-0.4cm}
    \ \\
    \left\lceil \frac{\max\{|b-b'|, \, |c-c'|\}}{2} \right\rceil \ \ \ ((b-b') \cdot (c-c') < 0)
\end{cases}, \nonumber
\end{eqnarray}
where $T(b,c)$ represents a table that can be formed as the above table.
\end{lemma}

The proofs are provided on our GitHub page.

For setting the value of $l \cdot \beta$ used in our smooth private selection, we pre-computed the values on the right-hand side of (5) and (8). The Python codes for the experiments are also provided on our GitHub page.

\subsection{Comparison with Global Sensitivity-based Methods}

We first compared the smooth private selection with the exponential mechanism and the permute-and-flip. Here, the distribution with density $h(z) = \frac{\sqrt{2}}{\pi (1+z^4)}$, that is, $\gamma = 4$, was used for our method based on Theorem 6. The evaluation of the impact of the variation in $\gamma$ is reported in the next subsection.

The simulation data for the experiments were generated as follows, with reference to existing studies \cite{17,18}:
\begin{eqnarray}
&& s[r] \, (:= b[r] + c[r]) = \mathrm{Binomial}(2N, 2/3), \nonumber\\
&& b[r] = \mathrm{Binomial}(s[r],1/2), \ \ c[r] = s[r] - b[r], \nonumber
\end{eqnarray}
where $N = 150$ and $r \in \mathcal{R} = \{1,2,\dots,m\}$. Here, $m$ is the number of SNPs. 

\subsubsection{Accuracy}

We first evaluated the accuracy while varying the value of $m$ between $5$, $10$, $15$, and $20$. As for our method, we considered the following four cases: (I) Theorem 4 is used to compute $S(x)$, and the noise is two-sided, (II) Theorem 4 is used to compute $S(x)$, and the noise is one-sided, (III) Theorem 5 is used to compute $S(x)$, and the noise is two-sided, and (IV) Theorem 5 is used to compute $S(x)$, and the noise is one-sided. The value of $\epsilon$ was varied from $3$ to $21$; the accuracy was measured over five iterations of $40$ runs. The results are shown in Fig. \ref{fig1}.

\begin{figure}[htbt]
  (a) \ \ \ \ \ \ \ \ \ \ \ \ \ \ \ \ \ \ \ \ \ \ \ \ \ \ \ \ \ \ \ \ \ (b)\\
  \centerline{
  \includegraphics[width=4.2cm]{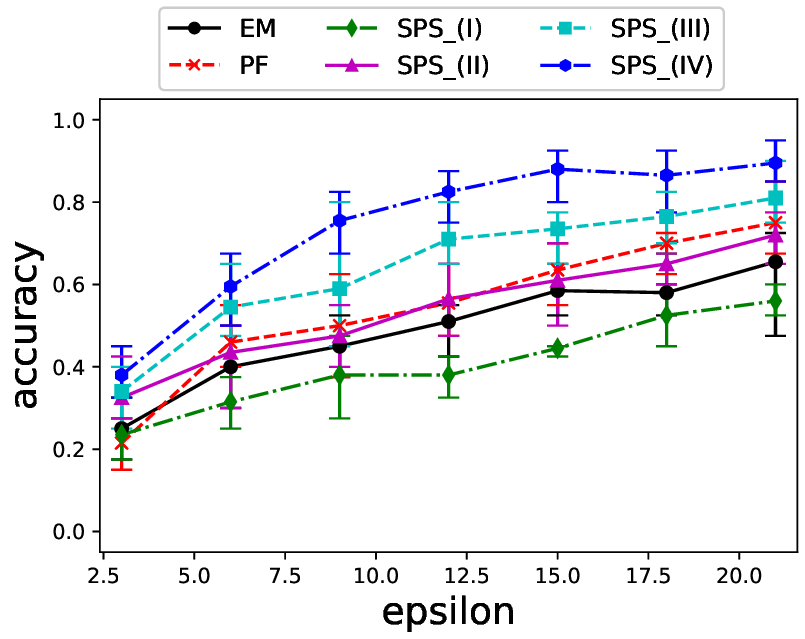}\ 
  \includegraphics[width=4.2cm]{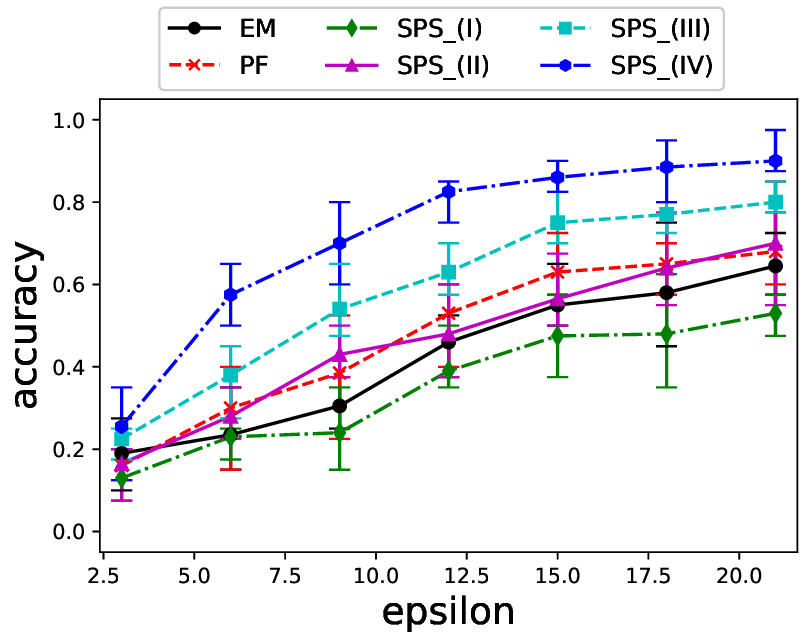}}
  
  (c) \ \ \ \ \ \ \ \ \ \ \ \ \ \ \ \ \ \ \ \ \ \ \ \ \ \ \ \ \ \ \ \ \ (d)\\
  \centerline{
  \includegraphics[width=4.2cm]{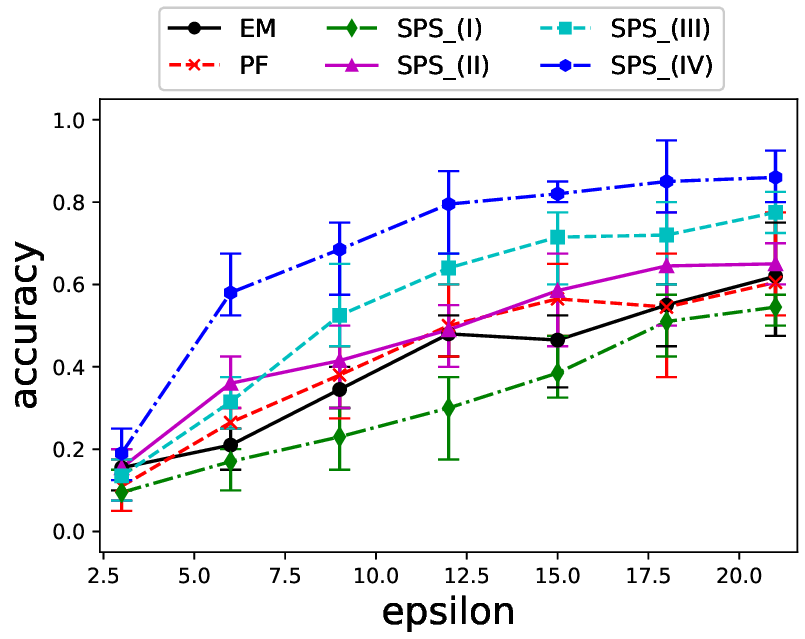}\ 
  \includegraphics[width=4.2cm]{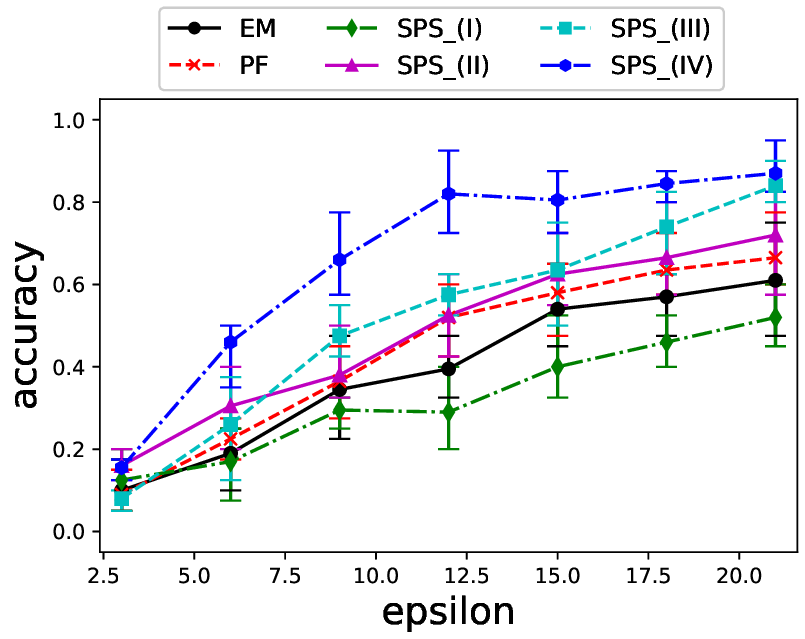}}
  \caption{Averaged accuracy when (a) $m = 5$, (b) $m = 10$, (c) $m = 15$, and (d) $m = 20$. The $x$-axis represents the achieved privacy level $\epsilon$. The $y$-axis represents the probability that the most significant SNP was correctly extracted. We compared the exponential mechanism (EM), the permute-and-flip (PF), and our smooth private selection (SPS). The error bar represents the range of all results in the five attempts.}
  \label{fig1}
\end{figure}

Fig. \ref{fig1} shows that our method outperformed the existing methods in Cases (III) and (IV) using Theorem 5, regardless of the value of $m$. When using Theorem 4 in Cases (I) and (II), the value of $l$ in our method must be set extremely small, so that the {\it smooth sensitivity} was not much different from the {\it global sensitivity}. Combined with the increase in noise due to the change in $\frac{1}{k\alpha}$, the accuracy was comparable to that of the existing {\it global sensitivity}-based methods. In contrast, using Theorem 5 can lower $S(x)$ because $l$ can be set relatively large; therefore, the added noise decreases. Furthermore, if the noise distribution is one-sided, the amount of perturbation required to satisfy $\epsilon$-differential privacy can be further reduced according to Theorems 2 and 3, which produced the highest accuracy in Case (IV). A comparison of the results for Cases (I) and (II) also implies that the use of one-sided noise rather than two-sided noise would be recommended for our method. 

\subsection{Effects of Variation in $\gamma$}

Next, we varied the value of $\gamma$ between $2$, $4$, $6$, and $10$ to investigate a suitable noise distribution for our method. The density functions $h(z)$ are $\frac{1}{\pi (1+z^2)}$, $\frac{\sqrt{2}}{\pi(1+z^4)}$, $\frac{3}{2\pi(1+z^6)}$, and $\frac{5\sqrt{5}-5}{4\pi(1+z^{10})}$, respectively. Here, we considered Case (IV), which had achieved the highest utility in the previous experiments.

\begin{table*}[htbt]
    \centering
    \caption{Run time (sec) of the mechanisms for data with $m$ elements.}
    \vspace{-0.2cm}
    \begin{tabular}{c||c|c|c|c|c|c}
        $m$ & $20$ & $50$ & $100$ & $200$ & $500$ & $1,000$ \\ \hline
        EM & $1.5 \times 10^{-4}$ & $3.5 \times 10^{-4}$ & $6.6 \times 10^{-4}$ & $1.5 \times 10^{-3}$ & $3.4 \times 10^{-3}$ & $6.4 \times 10^{-3}$ \\
        PF & $1.3 \times 10^{-4}$ & $3.2 \times 10^{-4}$ & $6.4 \times 10^{-4}$ & $1.4 \times 10^{-3}$ & $3.2 \times 10^{-3}$ & $6.3 \times 10^{-3}$ \\
        SPS$\_$Th.4 & $4.6 \times 10^{-4}$ & $1.1 \times 10^{-3}$ & $2.3 \times 10^{-3}$ & $4.8 \times 10^{-3}$ & $1.1 \times 10^{-2}$ & $2.0 \times 10^{-2}$ \\
        SPS$\_$Th.5 & $0.90$ & $2.2$ & $4.6$ & $10.0$ & $22.7$ & $43.1$ 
    \end{tabular}
    \label{table2}
    \vspace{-0.4cm}
\end{table*}

\subsubsection{Accuracy}

We evaluated the accuracy while varying the values of $m$ and $\epsilon$ as in Section V.A. The results are shown in Fig. \ref{fig3}.

\begin{figure}[htbt]
  (a) \ \ \ \ \ \ \ \ \ \ \ \ \ \ \ \ \ \ \ \ \ \ \ \ \ \ \ \ \ \ \ \ \ (b)\\
  \centerline{
  \includegraphics[width=4.2cm]{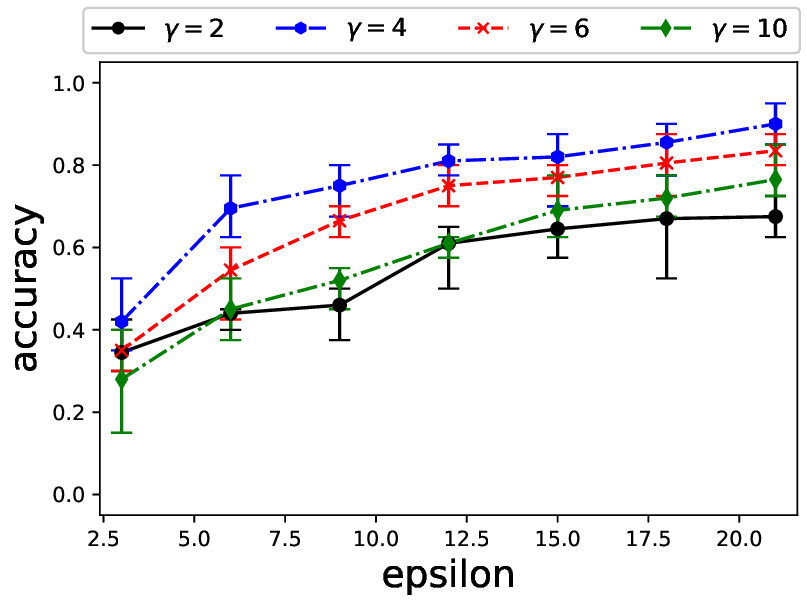}\ 
  \includegraphics[width=4.2cm]{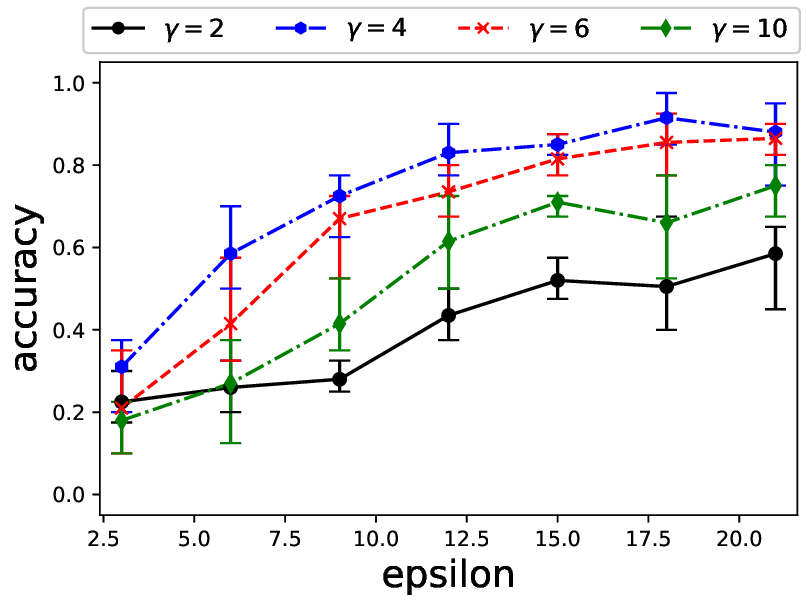}}

  (c) \ \ \ \ \ \ \ \ \ \ \ \ \ \ \ \ \ \ \ \ \ \ \ \ \ \ \ \ \ \ \ \ \ (d)\\
  \centerline{
  \includegraphics[width=4.2cm]{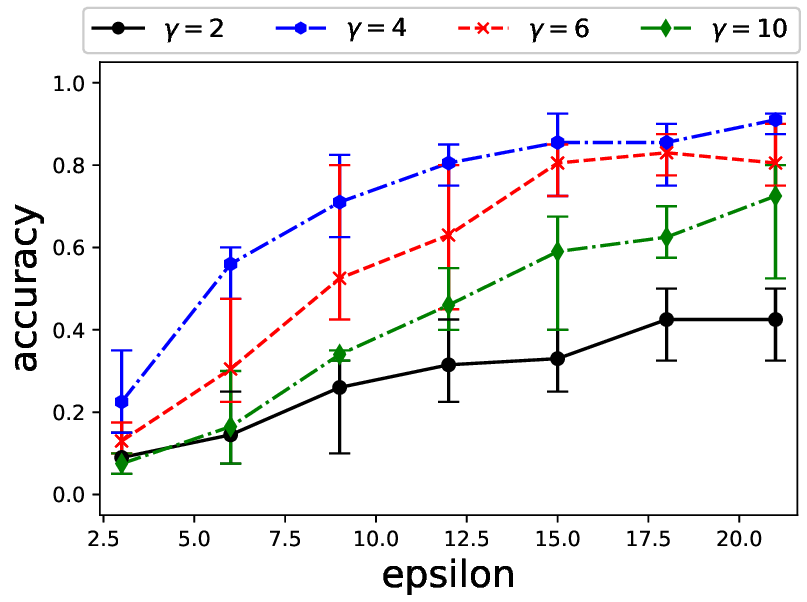}\  
  \includegraphics[width=4.2cm]{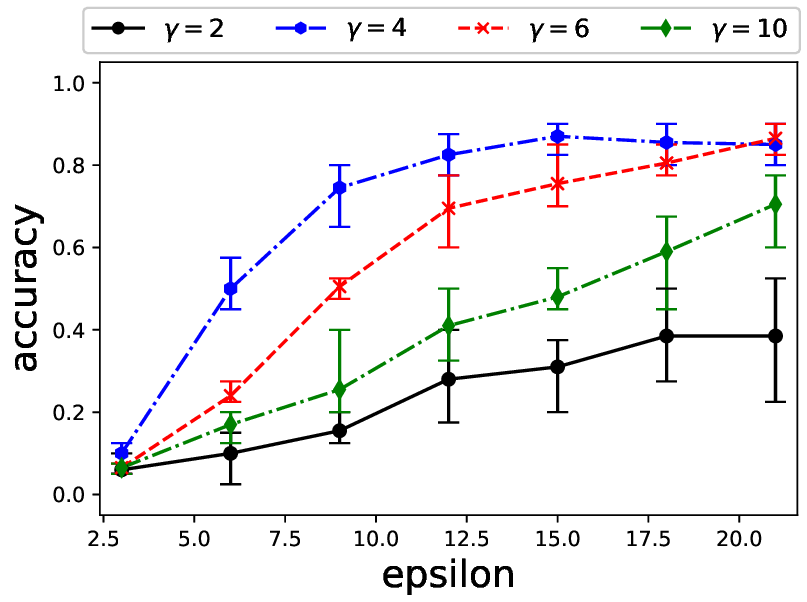}}
  \caption{Averaged accuracy when (a) $m = 5$, (b) $m = 10$, (c) $m = 15$, and (d) $m = 20$. The $x$-axis represents the achieved privacy level $\epsilon$. The $y$-axis represents the probability that the most significant SNP was correctly extracted. We compared the noise distributions with $\gamma = 2$, $4$, $6$, and $10$ using our smooth private selection. The error bar represents the range of all results in the five attempts.}
  \label{fig3}
\end{figure}

Fig. \ref{fig3} indicates that $\gamma = 4$ could be advisable for the smooth private selection. When $\gamma = 2$, the noise distribution is heavy-tailed, and a large noise is added with high probability; therefore, the accuracy of our mechanism would be poor. In contrast, when $\gamma \geq 4$, the mean and variance of the noise distribution can be defined, and having excessive noise is vary rare. In this case, the amount of noise will become larger as the value of $\gamma$ increases, mainly because the value of $\frac{1}{\alpha} = \frac{2(\gamma - 1)^\frac{\gamma-1}{\gamma}}{\epsilon}$ increases. Indeed, the accuracy was highest when $\gamma = 4$ in Fig. \ref{fig3}. \\

In this study, we briefly demonstrated that our smooth private selection has the potential to yield higher utility than the existing {\it global sensitivity}-based methods and evaluated a suitable value of $\gamma$ for the mechanism. 
However, in actual individual data analysis, various factors will influence the accuracy of the results, including the trend of $S(x)$, characteristics of the data distribution, and the possibility of other noise distributions. In particular, we expect that the error analysis of the smooth private selection for general cases would be presently too challenging, but it could be an important issue for the enrichment of differentially private mechanisms. Here, as an example, we show a condition under which the expected value of added noise in our smooth private selection is smaller than that in the permute-and-flip. 

\begin{lemma}
The expected value of added noise in the smooth private selection when using the distribution with density
\begin{eqnarray}
g(z) = \begin{cases}
    0 \ \ \ \ \ \ (z < 0)\\
    \frac{h(z)}{\int_{w \geq 0} h(w) \mathrm{d}w} \ \ (z \geq 0)
\end{cases}, \nonumber
\end{eqnarray}
where $h(z) = \frac{\sqrt{2}}{\pi (1+z^4)}$, is smaller than the expected value of added noise in the permute-and-flip if
\begin{eqnarray}
\max_{r \in \mathcal{R}} S(x,r) < \left( \frac{4}{27} \right)^{\frac{1}{4}} \cdot k \cdot GS_{u,\mathcal{R}}. \nonumber
\end{eqnarray}
\end{lemma}

\begin{proof}
Because
\begin{eqnarray}
\int_{-\infty}^{\infty} z \cdot g(z) \mathrm{d}z = \int_{0}^{\infty} z \cdot 2h(z) \mathrm{d}z = \frac{1}{\sqrt{2}}, \nonumber
\end{eqnarray}
the expected value of $\frac{S(x)}{k \cdot \alpha} \cdot Z_r$, where 
\begin{eqnarray}
S(x) := \max_{r \in \mathcal{R}} S(x,r), \nonumber
\end{eqnarray}
is
\begin{eqnarray}
&& \frac{2 \cdot 3^{\frac{3}{4}} \cdot S(x)}{k \cdot \epsilon} \cdot \frac{1}{\sqrt{2}} = \frac{2^{\frac{1}{2}} \cdot 3^{\frac{3}{4}}}{k \cdot \epsilon} \cdot S(x). \nonumber\\
&& \ \ \ \ \ \ \ \ \ \ \ \  \left[\because \ \alpha = \frac{\epsilon}{2 \cdot (\gamma - 1)^{\frac{\gamma - 1}{\gamma}}} = \frac{\epsilon}{2 \cdot 3^{\frac{3}{4}}} \right] \nonumber
\end{eqnarray}
In contrast, in the permute-and-flip, the expected value of $\mathrm{Expo}_r \left( \frac{\epsilon}{2 \cdot GS_{u,\mathcal{R}}} \right)$ is $\frac{2 \cdot GS_{u,\mathcal{R}}}{\epsilon}$. Therefore, if
\begin{eqnarray}
&& \frac{2^{\frac{1}{2}} \cdot 3^{\frac{3}{4}}}{k \cdot \epsilon} \cdot S(x) < \frac{2 \cdot GS_{u,\mathcal{R}}}{\epsilon} \nonumber\\
&\iff& S(x) < \left( \frac{4}{27} \right)^{\frac{1}{4}} \cdot k \cdot GS_{u,\mathcal{R}}, \nonumber
\end{eqnarray}
the expected value of the added noise is smaller for the smooth private selection.
\end{proof}

A similar analysis can be conducted in other cases or when varying the value of $\gamma$. Thereafter, with reference also to our Theorems 2 and 3, we can examine a relation among $\epsilon$, $|\mathcal{R}|$, and the number of individuals, for which higher accuracy may be provided by the smooth private selection. However, this is not sufficient for the general error analysis. Particularly, the relationship between $k$ and $l$, that is, $k \cdot \alpha$ and $\max_r S(x,r)$, should be given further attention. Depending on the characteristics of data, the optimal values should differ that can minimize the added noise. We intend to advance this discussion and explore the settings of $k$ and $l$ in the future, even for deepening our Theorem 1.

\subsection{Run Time}

We further measured the run time of the mechanisms (the exponential mechanism (EM), the permute-and-flip (PF), and our smooth private selection (SPS)) while varying the value of $m$ from $20$ to $1,000$. As for our method, the noise distribution being one-sided or two-sided minimally affects the run time; therefore, we considered the two cases of using Theorems 4 and 5 to compute $S(x)$. The results over $10$ runs are shown in Table \ref{table2}.

Table \ref{table2} shows that the run time is almost proportional to $m$. Our method requires $S(x)$ in addition to the {\it global sensitivity}; therefore, the run time exceeds that of the existing methods. In particular, using Theorem 5 is more time-consuming than using Theorem 4, because the distance to $U$ that truly contains the set of $x$ such that $LS_f(x) = GS_f$ must be computed for each element. However, even when $m = 1,000$, our smooth private selection takes $< 1$ minute, indicating that it can be sufficiently fast to outperform the existing methods regarding utility, combined with the previous results. As with the accuracy, the run time inherently depends on the dataset and analysis purpose, but the results show that our method can maintain acceptable processing times.

\subsection{Example based on Real Data}
\vspace{-0.3cm}
\begin{table}[htbt]
    \centering
    \caption{The probability (\%) that the most significant SNP was correctly extracted.}
    \vspace{-0.1cm}
    \begin{tabular}{c||c|c|c|c|c|c|c}
        $\epsilon$ & $3$ & $6$ & $9$ & $12$ & $15$ & $18$ & $21$ \\ \hline
        EM & $30.0$ & $42.5$ & $42.5$ & $55.0$ & $60.0$ & $72.5$ & $82.5$  \\
        PF & $22.5$ & $40.0$ & $47.5$ & $55.0$ & $62.5$ & $77.5$ & $80.0$ \\
        SPS (I) & $12.5$ & $22.5$ & $27.5$ & $40.0$ & $45.0$ & $50.0$ & $67.5$ \\
        SPS (II) & $25.0$ & $42.5$ & $47.5$ & $60.0$ & $60.0$ & $70.0$ & $87.5$ \\
        SPS (III) & $15.0$ & $37.5$ & $42.5$ & $45.0$ & $50.0$ & $60.0$ & $82.5$ \\
        SPS (IV) & $\mathbf{32.5}$ & $\mathbf{55.0}$ & $\mathbf{65.0}$ & $\mathbf{67.5}$ & $\mathbf{77.5}$ & $\mathbf{82.5}$ & $\mathbf{92.5}$
    \end{tabular}
    \label{table1}
    \vspace{-0.1cm}
\end{table}

Finally, we evaluated the accuracy of our method based on a real result of TDT for $215$ families \cite{20}, with reference to the existing study \cite{18}. Using the six statistics provided in the paper, we measured the probability of retrieving the most significant SNP. We set $\gamma$ as $4$ for the smooth private selection, and the cases from (I) to (IV) are identical to those in Section V.A. The results over $40$ runs are shown in Table \ref{table1}.

The results in Table \ref{table1} show that for all $\epsilon$, our method using Theorem 5 and one-sided noise yielded the highest accuracy. The reason for the accuracy being higher in Case (II) than in Case (I) and in Case (IV) than in Case (III) is that the value of $l$ can be smaller for one-sided distribution than for two-sided distribution. The reason for the accuracy being higher in Case (III) than in Case (I) and in Case (IV) than in Case (II) is that Theorem 5 can provide a smaller $S(x)$ than Theorem 4. These results are consistent with those in Section V.A.
\section{Conclusion}

In this study, we proposed the first mechanism for differentially private selection using {\it smooth sensitivity}, the smooth private selection, along with theoretical proofs of the privacy guarantees. We also presented new fundamental algorithm and theorems on {\it smooth sensitivity} and expanded its potential. 
The experimental evaluation showed that our proposed mechanism can provide higher accuracy than the existing {\it global sensitivity}-based methods in privacy-preserving data analysis. Important future challenges include the following:
\begin{enumerate}
\item[i)] Developing an optimal method for determining the values of $k$ and $l$ that provide the highest accuracy for each analysis data;

\item[ii)] Conducting general theoretical analysis of our mechanism given the characteristics of the data, associated with the first challenge;

\item[iii)] Exploring possible noise distributions with a density function other than $h(z) \propto \frac{1}{1+|z|^\gamma}$;

\item[iv)] Integrating our mechanism with the joint approach \cite{14} and the local dampening mechanism \cite{16}, while developing efficient algorithms and more generalized theorems for computing {\it smooth sensitivity}. 
\end{enumerate}

\bibliographystyle{plain}
\bibliography{mybibliography}


\end{document}